\documentclass[a4paper]{article}
\usepackage{chngcntr}
\usepackage{hyperref}

\usepackage{algorithmicx}
\usepackage[utf8]{inputenc}
\usepackage{algorithm}
\usepackage{algpseudocode}
\usepackage[margin=2.7cm]{geometry}
\usepackage{graphicx}
\usepackage{float}
\usepackage{subfig}
\usepackage{amsthm,color}
\usepackage{amsmath,amssymb}
\usepackage[dvipsnames]{xcolor}
\usepackage{paralist}
\usepackage{authblk}

\newcommand{\rmv}[1]{}

\newcommand{\RMP}[1]{}

\newcommand{\XMP}[1]{}

\newcommand{\ars}{{registers}}

\newcommand{\W}[1]{\textsc{write}(#1)}

\newcounter{NewCounter}

\newtheorem{theorem}[NewCounter]{Theorem}

\newtheorem{claim}{Claim}[NewCounter]

\newtheorem{corollary}[NewCounter]{Corollary}
\newtheorem{lemma}[NewCounter]{Lemma}
\newtheorem{observation}[NewCounter]{Observation}

\newtheorem{definition}[NewCounter]{Definition}
\newtheorem{notation}[NewCounter]{Notation}

\algdef{SE}[SUBALG]{Indent}{EndIndent}{}{\algorithmicend\ }
\algtext*{Indent}
\algtext*{EndIndent}

\title{Randomized Consensus with Regular Registers}

\author[1]{Vassos Hadzilacos}
\author[1]{Xing Hu}
\author[1]{Sam Toueg}

\affil[1]{Department of Computer Science, University of Toronto}

\begin{document}
\maketitle

\begin{abstract}
The well-known randomized consensus algorithm by Aspnes and Herlihy~\cite{AH1990consensus}
	for asynchronous shared-memory systems was
	proved to work, even against a strong adversary,
	under
	the assumption that the registers that it uses are \emph{atomic} registers.
With atomic registers every read or write operation is \emph{instantaneous} (and thus indivisible).
As pointed out in~\cite{sl1}, however,
	a randomized algorithm that works with atomic registers
	does not necessarily work
	if we replace the atomic registers that it uses with linearizable implementations of registers.
	
This raises the following question:
	does the randomized consensus algorithm by Aspnes and Herlihy still work against a strong adversary
	if we replace its atomic registers
	with linearizable registers?

We show that the answer is affermative,
	in fact we show that even linearizable registers are not necessary.
More precisely,
	we prove that the algorithm by Aspnes and Herlihy works
	against a strong adversary even if 
	the algorithm uses only \emph{regular} registers.

%
%
%

\end{abstract}

\section{Introduction}

In the \emph{consensus problem}, each process \emph{proposes} some value 
	and must \emph{decide} a value such that the following properties hold:
	
\begin{itemize}
	\item \textbf{Validity}: If a process decides a value $v$ then some process proposes $v$.
	\item \textbf{Agreement}: No two processes decide different values.
	\item \textbf{Termination}: Every non-faulty process eventually decides a value.
\end{itemize}

This problem cannot be solved
	in asynchronous shared-memory systems with process crashes~\cite{AbuAmara},
	but there are \emph{randomized} algorithms that solve
	a weaker version of the consensus problem that
	requires Termination ``only'' with probability 1.
In particular, the well-known randomized consensus algorithm by Aspnes and Herlihy~\cite{AH1990consensus}
	for shared-memory systems was
	proved to work, even against a strong adversary and with any number of process crashes, under
	the assumption that the registers that it uses are \emph{atomic} Single-Writer Multiple-Reader (SWMR) registers.
With atomic registers every read or write operation is \emph{instantaneous} (and thus indivisible).

As pointed out in~\cite{sl1}, however,
	a randomized algorithm that works with atomic registers
	does not necessarily work (i.e., it may lose some of its properties, including termination)\XMP{Really?}
	if we replace the atomic registers that it uses
	with \emph{linearizable (implementations of) $\ars$};\footnote{Roughly speaking,
	an object (implementation) is linearizable~\cite{HerlihyWing1990}
	if, although each object operation spans an interval between its invocation and its response,
	and so operations can be \emph{concurrent},
	operations behave as if they occur in a sequential order (called ``linearization order'')
	that is consistent with the order in which operations actually occur: if an operation $o$ completes
	before another operation $o'$ starts, then $o$ precedes $o'$ in the linearization order.}
	intuitively, this is because a strong adversary can exploit a weakness in the linearizability requirement
	to break some properties of the algorithm.
If, however, we replace the algorithm's atomic registers
	with \emph{strongly} linearizable (implementations of) $\ars$,
	then the algorithm is guaranteed to retain its original properties~\cite{sl1}.
This leads to the following natural question:
	does the randomized consensus algorithm by Aspnes and Herlihy still work against a strong adversary
	if we replace the atomic registers that it uses
	with linearizable registers,
	or does it require strongly linearizable registers?
	
\emph{Prima facie}, it appears that this algorithm requires strongly linearizable registers to work.
This is because, as we explain in the Appendix,
	the proof given in~\cite{AH1990consensus} that the algorithm terminates with \mbox{probability~1}
	relies on the fact that the algorithm's registers are atomic:
	this proof does \emph{not} work with (non-atomic) linearizable registers.

We show here that strong linearizability is not necessary for this algorithm to work.
In fact, we show the perhaps surprising result that even linearizable registers are not necessary.
More precisely,
	we prove that the algorithm works
	against a strong adversary even if 
	the algorithm uses only \emph{regular} SWMR registers.

\section{Regular registers}

Each read or write operation on a non-atomic register (such as a regular register) spans an interval that starts with an \emph{invocation} and terminates with a \emph{response}.

\begin{definition}
Let $o$ and $o'$ be any two register operations.

\begin{itemize}

\item $o$ \emph{precedes} $o'$ if the response of $o$ occurs before the invocation of $o'$.
	
\item $o$ \emph{is concurrent with} $o'$ if neither precedes the other.

\end{itemize}
\end{definition}

\begin{definition}

A register initialized to some value $v_0$ is a \emph{regular} register if and only if it satisfies the following property.
If the response of a read operation $r$ returns the value $v$ then:
\begin{enumerate}

\item there is a write $v$ operation that immediately precedes $r$ or is concurrent~with~$r$, or	
\item no write operation precedes $r$ and $v = v_0$.
\end{enumerate}

\end{definition}

Note that regular registers are \emph{not} linearizable registers because they allow ``new-old inversions'': two consecutive read operations that are concurrent with a write operation can be such that the first read returns the new value of the register, while the second read returns the old value.

\section{Aspnes and Herlihy's algorithm works with regular registers}

In the algorithm by Aspnes and Herlihy (shown in Algorithm~\ref{Algocon}),
	every process $p$ now has a \emph{regular} SWMR
	register $R[p]$ that $p$ can write and every process can read.
The register $R[p]$ contains a pair $(prefer,round)$:
	$R[p].prefer$ is either a value in $\{0,1\}$ that process $p$ currently ``prefers''
	(and will eventually be the value that $p$ decides),
	or the special symbol $\bot$ (which indicates that $p$ is temporarily ``paused'');
	$R[p].round$ is a counter that increases in each iteration
	(except for each iteration where $p$ is paused, i.e, $R[p].prefer$ is $\bot$).

At the beginning of the algorithm, each process $p$ sets $R[p]$ to contain the pair $(v,1)$,
	where $v$ is the value that $p$ proposes. 
The process then enters a loop in which it first reads the registers of all processes,
	and based on what it reads it determines whether to decide and, if not, how to update its register.
Given the values of the registers that a process $p$ reads at the start of each iteration,
	the following terms are used in the description of the algorithm and its proof:


\begin{definition} ~
\begin{itemize}
\item  A process $p$ is a \emph{leader} if $R[p].round \geq R[j].round$ for every process $j$. 

\item A \emph{process $q$ agrees with $p$} if
	$R[p].prefer  =  R[q].prefer\neq \bot$.
	
\item A \emph{process $q$ agrees with $p$ on some value $v$}
	if $v \neq \bot$ and $R[p].prefer  =  R[q].prefer = v$.

\item A \emph{process $q$ trails $p$ by at least 2 rounds} if $R[p].round \ge R[q].round+2$.

\end{itemize}

\end{definition}

\begin{algorithm}[!h]
\caption{ Aspnes and Herlihy's Randomized Consensus Algorithm}\label{Algocon}
~\\
For each process $p \in \Pi$: $R[p]$ is a shared regular SWMR register that $p$ can write\newline
\hspace*{3.8cm}and every process can read;
	 initially $R[p] = (\bot,0)$ 

\vspace{2mm}

\begin{algorithmic}[1]
\Statex
\textsc{Consensus($v$):} \Comment{code executed by process $p$ to propose the value $v$}
\Indent
\State $R[p] \leftarrow ( v, 1)$ \Comment{$\W{v,1}$}
\State \textbf{while true do} 
\Indent
\State read all registers $R[*]$ \label{readall}
\State $(x,r) \leftarrow R[p]$\label{capture} \Comment{$x \in \{0, 1, \bot \}$}
\State \textbf{if} I'm a leader \textbf{and} all who disagree (with me) trail by at least 2 rounds:\label{notdecide}
\Indent 
\State decide $x$\label{decide} and halt\Comment{$x \in \{0, 1 \}$}
\EndIndent
\State \textbf{else if} leaders agree on some value $v^*$:\label{lagree}
\Indent 
\State $R[p] \leftarrow ( v^*,  r+1  )$\label{aleader} \Comment{$\W{v^*,  r+1}$}
\EndIndent
\State \textbf{else if} $x \neq \bot$: \label{checkbot}
\Indent 
\State $R[p] \leftarrow ( \bot, r  )$\label{bot} \Comment{$\W{\bot, r }$}
\EndIndent
\State \textbf{else}:
\Indent 
\State $R[p] \leftarrow ( \text{flip()},  r +1 )$\label{acoin} \Comment{$\W{\text{flip()},  r +1}$}
\EndIndent
\EndIndent
\EndIndent
\end{algorithmic}
\end{algorithm}

Unless we indicate otherwise, 
	henceforth we consider an arbitrary execution of Algorithm~\ref{Algocon}.
In the following we say that a process \emph{crashes} if it does not halt but takes finitely many steps, and 
	a process is \emph{correct} if it does not crash.

\newpage
\begin{notation} ~
\begin{itemize}
\item A process $q$ \emph{invokes $\W{x,  y}$}, if $q$ invokes an operation to write $(x,y)$ into~register~$R[q]$.
\item $v$ and $\bar{v}$ are values in $\{ 0,1\}$ such that $\bar{v} =1 -v$.
\end{itemize}
\end{notation}

\begin{observation}\label{observation0}
For every round $r \ge 1$:

\begin{enumerate}[(a)]

\item\label{zero} If a process invokes $\W{x ,r}$ then $x \in \{ 0,1, \bot\}$.

\item\label{uno1} For all $x \in \{ 0,1, \bot\}$, a process invokes $\W{x,r}$ at most once.

\item\label{uno2} No process invokes both $\W{0,r}$ and $\W{1,r}$.

\item\label{mono} If a process invokes $\W{- ,r}$ before it invokes $\W{- ,r'}$ for some $r'$ then $r \le r'$.

\item\label{due} If a process invokes $\W{\bot,r}$, then it invokes $\W{v,r}$\RMP{actually it completes $\W{v,r}$}
	for some $v \in \{0,1 \}$ before it invokes $\W{\bot,r}$.
\end{enumerate}

\end{observation}

\begin{lemma}\label{bl1}\RMP{Provide intuition to some key lemmas}
For all $r \ge 1$, if a process $p$ invokes $\W{v,r}$ and then invokes \mbox{$\W{\bar{v},r+1}$},\linebreak
	then some
	process $q \neq p$ invokes $\W{\bar{v},r'}$ with $r'\geq r$ before $p$ invokes $\W{\bar{v},r+1}$.
 \end{lemma}
\begin{proof}
Let $r \ge 1$ and suppose a process $p$ invokes $\W{v,r}$ and then invokes $\W{\bar{v},r+1}$.
Consider the while iteration in which $p$ invokes $\W{\bar{v},r+1}$.
Note that $p$ completes its $\W{v,r}$ operation before it reads the registers $R[-]$ in line~\ref{readall} of that iteration.
Thus, since $R[p]$ is a regular register,
    $p$ reads $(-,r)$ from $R[p]$ in line~\ref{readall}.
    
Let $r'$ be the round of the leaders that $p$ sees in line~\ref{readall}, i.e.,
	$p$ reads $(-,r')$ from the register of every leader in line~\ref{readall}.
Since $p$ reads $(-,r)$ from $R[p]$, by the definition of leaders, $r' \geq r$.

\begin{claim}\label{bc0}
$p$ reads either $(\bar{v},r')$ or $(\bot,r')$ from the register $R[\ell]$ of some leader $\ell$.
\end{claim}

\begin{proof}
Note that $p$ invokes $\W{\bar{v},r+1}$ by executing either line~\ref{aleader} or line~\ref{acoin}.

\begin{itemize}
\item \textbf{Case 1:} If $p$ executes line~\ref{aleader}, then
	by the condition of line~\ref{lagree}, $p$ must have seen that the leaders agree on $\bar{v}$,
	so, $p$ reads $(\bar{v},r')$ from all the leaders.
	
\item \textbf{Case 2:} if $p$ executes line~\ref{acoin},
	then, by the condition of line~\ref{lagree},
	it must have seen that the leaders do not agree.
So, either $p$ reads $(\bot,r')$ from the register of at least one leader,
	or $p$ reads both $(v,r')$ and $(\bar{v},r')$ from the registers of some leaders.

\end{itemize}

So in both cases the claim holds.
\end{proof}

Since $R[\ell]$ is a regular register, by Claim~\ref{bc0},
    $\ell$ invokes $\W{\bar{v},r'}$ or $\W{\bot,r'}$ before $p$ completes reading $R[\ell]$.
So $\ell$ invokes $\W{\bar{v},r'}$ or $\W{\bot,r'}$ before 
     $p$ invokes $\W{\bar{v},r+1}$.

\textbf{Case 1:}
 $\ell$ invokes $\W{\bar{v},r'}$ before $p$ invokes $\W{\bar{v},r+1}$.
We claim that $\ell \neq p$.
Suppose for contradiction that $\ell =p$.
Since $p$ invokes $\W{\bar{v},r'}$ before $p$ invokes $\W{\bar{v},r+1}$, and $r' \ge r$,
	by Observation~\ref{observation0}(\ref{uno1}) and~(\ref{mono}),
	$r' = r$.
So $p$ invokes both invokes $\W{\bar{v},r}$ and $\W{v,r}$ --- a contradiction to Observation~\ref{observation0}(\ref{uno2}).
%
%
%
%
%
Since $\ell$ invokes $\W{\bar{v},r'}$ before $p$ invokes $\W{\bar{v},r+1}$ and $\ell \neq p$, the lemma holds for~$q=\ell$.

\textbf{Case 2:} 
   $\ell$ invokes $\W{\bot,r'}$ before $p$ invokes $\W{\bar{v},r+1}$.
Let $f$ be the \emph{first} process that invokes $\W{\bot,r'}$; 
	recall that $r' \geq r$.
Note that $f$ invokes $\W{\bot,r'}$ in line~\ref{bot} of some while iteration.
Let $r^*$ be the round of the leaders that $f$ sees in line~\ref{readall}, i.e.,
	$f$ reads $(-,r^*)$ from the register of every leader in line~\ref{readall} of that iteration.
Since $f$ invokes $\W{\bot,r'}$ in line~\ref{bot},
	$f$ reads $(-,r')$ from $R[f]$ in line~\ref{readall}. 
So, by definition of leaders, the leaders that $f$ sees in line~\ref{readall} have $r^* \geq r'$; since $r' \geq r$, we have $r^* \geq r$.

Note that $f$ does not read $(\bot,r^*)$ from any leader's register in line~\ref{readall}: this is because,
	since the shared registers are regular, 
	that leader would have invoked $\W{\bot,r^*}$
	with $r^* \geq r$ before $f$ invokes $\W{\bot,r'}$ with $r' \geq r$ --- contradicting the definition of $f$.

Since $f$ invokes $\W{\bot,r'}$ in line~\ref{bot},
	by the condition of line~\ref{lagree},
 	$f$ sees that the leaders do not agree.
Since $f$ does not read $(\bot,r^*)$ from any leader's register in line~\ref{readall} but it sees that
	the leaders do not agree,
	$f$ must read both $(v, r^*)$ and $(\bar{v}, r^*)$ from some leaders' registers.
Let $q$ be a leader such that $f$ reads $(\bar{v}, r^*)$ from $R[q]$ in line~\ref{readall}.
Since $R[q]$ is a regular register,
	$q$ invokes $\W{\bar{v},r^*}$ before $f$ completes the reading of $R[q]$.
Thus, $q$ invokes $\W{\bar{v},r^*}$ before $f$ invokes $\W{\bot,r'}$ in line~\ref{bot}.
By the choice of $f$,
    $q$ invokes $\W{\bar{v},r^*}$ before $\ell$ invokes $\W{\bot,r'}$.
So, by the hypothesis of Case~2,
    $q$ invokes $\W{\bar{v},r^*}$ before $p$ invokes $\W{\bar{v},r+1}$.
Note that $q \neq p$ (the proof is the same as the proof of $\ell \neq p$ in Case 1).
Thus, since $r^*\geq r$, the lemma holds for $r' = r^*$.
\end{proof}

The next lemma generalizes Lemma~\ref{bl1}.

\begin{lemma}\label{bl1.5}
For all $r\ge1$, if a process $p$ invokes $\W{v,r}$ before it invokes $\W{\bar{v},r'}$ with $r' >r$,
	then some
	process $q \neq p$ invokes $\W{\bar{v},r''}$ with $r''\geq r$ before $p$ invokes $\W{\bar{v},r'}$.
 \end{lemma}
 
\begin{proof}
Let $r\ge1$ and suppose a process $p$ invokes $\W{v,r}$ and then invokes $\W{\bar{v},r'}$ with $r' >r$.
From Observation~\ref{observation0}(\ref{due}),
	for every round $j$, $r \le j \le r'$,
	$p$ invokes $\W{v,j}$ or $\W{\bar{v},j}$.
So there is a round $\hat{r}$, $r \le \hat{r} < r'$,
	such that $p$ invokes $\W{v,\hat{r}}$ and then invokes $\W{\bar{v},\hat{r}+1}$.
By Lemma~\ref{bl1},
	some process $q \neq p$ invokes $\W{\bar{v},r''}$ with $r''\geq \hat{r}$ before $p$ invokes $\W{\bar{v},\hat{r}+1}$.
Since $\hat{r} \ge r$, $r'' \ge r$, and since $\hat{r} < r'$, $p$ invokes $\W{\bar{v},\hat{r}+1}$
	no later than when it invokes $\W{\bar{v},r'}$.
So $q \neq p$ and $q$ invokes $\W{\bar{v},r''}$ with $r''\geq r$ before $p$ invokes $\W{\bar{v},r'}$.
\end{proof}

 \begin{lemma}\label{bl2}
For all $r \ge 1$,
	if no process invokes $\W{\bar{v},r}$ by some time in the execution of the algorithm,
	no process invokes $\W{\bar{v},r'}$ for any $r'\geq r$ by that time.
\end{lemma}

 \begin{proof}
Assume,
    for contradiction,
    there is a round $r \geq 1$ and a time $t$ such that
    no process invokes $\W{\bar{v},r}$ by time $t$,
	but some processes invoke $\W{\bar{v},r'}$ with $r'\geq r$ by time $t$. 
Let $p$ be the \emph{first} process that invokes $\W{\bar{v},r'}$ with $r'\geq r$; this must occur by time $t$.
Since no process invokes $\W{\bar{v},r}$ by time $t$,
	$r' > r$.
Note that $p$ invokes $\W{-,r}$ before $p$ invokes $\W{\bar{v},r'}$.
So, by Observation~\ref{observation0}(\ref{due}),
	$p$ invokes $\W{w,r}$ with a value $w \in \{ v,\bar{v} \}$ before it invokes $\W{\bar{v},r'}$ by time $t$.
 Since no process invokes $\W{\bar{v},r}$ by time $t$,
	$w = v$.
By Lemma \ref{bl1.5},
    some process $q \neq p$ invokes $\W{\bar{v},r''}$ with $r'' \geq r$ before $p$ invokes $\W{\bar{v},r'}$
    --- a contradiction since $p$ is the \emph{first} process that invokes $\W{\bar{v},r'}$ with $r' \geq r$ by time $t$.
\end{proof}

We now prove a similar lemma for the special value $\bot$.

 \begin{lemma}\label{bl2'}
 For all $r \ge 1$,
	if no process invokes $\W{\bar{v},r}$  by some time in the execution of the algorithm,
	no process invokes $\W{\bot,r'}$ for any $r'\geq r$ by that time.
\end{lemma}

\begin{proof}
Assume,
    for contradiction,
    there is a round $r \geq 1$ and a time $t$ such that
    no process invokes $\W{\bar{v},r}$ by time $t$,
    but some process invokes $\W{\bot,r'}$ with $r'\geq r$ by time $t$. 
Let $p$ be the \emph{first} process that invokes $\W{\bot,r'}$ with $r'\geq r$ by time $t$.
Note that $p$ does so in line~\ref{bot} for some while iteration.
Let $r''$ be the round of the leaders that $p$ sees in line~\ref{readall}, i.e.,
	$p$ reads $(-,r'')$ from the register of every leader in line~\ref{readall} of that iteration.
Since $p$ invokes $\W{\bot,r'}$ in line~\ref{bot},
	$p$ reads $(-,r')$ from $R[p]$ in line~\ref{readall} by time $t$.
So, by definition of leaders, the leaders that $p$ sees in line~\ref{readall} have $r'' \geq r'$;
	since $r' \geq r$, we have $r'' \geq r$.

Since no process invokes $\W{\bar{v},r}$ by time $t$, and $r'' \geq r$,
    by Lemma \ref{bl2},
    no process invokes $\W{\bar{v},r''}$ by time $t$.
So, since the shared registers are regular, 
	$p$ cannot read $(\bar{v},r'')$ from any register $R[-]$ in line~\ref{readall} by time $t$.

Since $p$ invokes $\W{\bot,r'}$ in line~\ref{bot},
	by the condition of line~\ref{lagree},
    $p$ must see that the leaders do not agree.
So, since $p$ does not read $(\bar{v},r'')$ from any register $R[-]$ in line~\ref{readall},
    $p$ sees the leaders do not agree because it reads $(\bot,r'')$
    from the register $R[\ell]$ for at least one leader $\ell$ by time $t$.
Since $R[\ell]$ is a regular register,
    $\ell$ invokes $\W{\bot,r''}$ before $p$ completes its reading of $R[\ell]$ in line~\ref{readall}.
So,
    $\ell$ invokes $\W{\bot,r''}$ with $r'' \geq r$ before $p$ invokes $\W{\bot,r'}$ in line~\ref{bot}, and $\ell$ does so by time $t$ --- a contradiction since $p$ is the \emph{first} process that invokes $\W{\bot,r'}$ with $r'\geq r$ by time $t$.
\end{proof}

By Lemmas~\ref{bl2} and \ref{bl2'}, and Observation~\ref{observation0}(\ref{zero}), we have the following:

\begin{corollary}\label{bc2.2}
For all $r \ge 1$,
	if no process invokes $\W{\bar{v},r}$,\RMP{by some time t in the run}
	then for all $r'\geq r$,
	every process that invokes $\W{-,r'}$\RMP{by time t in the run}
	invokes $\W{v,r'}$.
\end{corollary}


\begin{definition}\label{def-decide}
If a process $p$ executes line~\ref{decide} when $R[p] = (v,r)$, we say that $p$ decides~$v$ at round $r$.
\end{definition}

 We first show that the algorithm satisfies the \textbf{Validity} property.  

\begin{lemma}[Validity]\label{validity}
If a process decides $v$ then some process proposes $v$.
\end{lemma}

\begin{proof}
Suppose, for contradiction, that some process $p$ decides $v$
	but all the processes that propose a value, propose $\bar{v}$.
So they all invoke $\W{\bar{v},1}$, and by Observation~\ref{observation0}(\ref{uno2}) no process invokes $\W{v,1}$.
By Corollary~\ref{bc2.2}, for all $r\geq 1$,
	every process that invokes $\W{-,r}$
	invokes $\W{\bar{v},r}$.
So the register $R[p]$ can only contain the value $\bar{v}$ or $\bot$.
Note that when $p$ decides (this occurs in line~\ref{decide}),
	$p$ decides a value $x$ only if $x$ is in $R[p]$ and $x \neq \bot$.
Thus $p$ can only decide $\bar{v}$ --- a contradiction.
\end{proof}

We now proceed to show that the algorithm satisfies the \textbf{Agreement} property.

\begin{lemma}\label{bl3.5}
For all round $r\geq 1$, if some process $p$ decides $v$ at round $r$,
    no process invokes $\W{\bar{v},r}$.
\end{lemma}

\begin{proof}
For $r=1$:
	Suppose
	 some process $p$ decides $v$ at round $1$.
By Definition~\ref{def-decide}, this means that $p$ has $R[p] = (v,1)$ when it executes line~\ref{decide}.
Since initially $R[p] = (\bot, 0)$,
    by the condition of line \ref{notdecide},
    $p$ reads $(v,1)$ from the register of every process in line~\ref{readall}.
Since the registers are regular,
    every process invokes $\W{v,1}$.
By Observation~\ref{observation0}(\ref{uno2}),
    no process invokes $\W{\bar{v},1}$.
Thus,
    by Corollary~\ref{bc2.2},
    for all $r \ge 1$,  no process invokes $\W{\bar{v},r}$.
    
For $r\geq 2$:
	Suppose, for contradiction,
    that some process $p$ decides $v$ at round $r$,
    but some process invokes $\W{\bar{v},r}$.
Let $q$ be the \emph{first} process that invokes $\W{\bar{v},r}$.

Since $p$ decides $v$ at round $r$ in some while iteration,\RMP{Name this iteration $I_p$ and use it later?}
	by Definition~\ref{def-decide},
	$p$ executes line~\ref{decide} when $R[p] = (v,r)$; so $p$ reads $(v,r)$ from $R[p]$ in line~\ref{readall} in that iteration.
So, since $R[p]$ is a regular register,
	it is clear that the following holds:

\begin{claim}\label{ob2}
$p$ completes $\W{v,r}$ before $p$ starts reading registers in line~\ref{readall} of the while iteration where $p$ decides $v$ at round $r$.
\end{claim}

Since $r \ge 2$, $q$ invokes $\W{\bar{v},r}$ by executing  either (1) line~\ref{aleader} or (2) line~\ref{acoin}.
We now prove that both cases are impossible (and so the lemma holds):
\begin{itemize}

\item
\textbf{Case 1}: $q$ invokes $\W{\bar{v},r}$ by executing line~\ref{aleader}.
Consider the while iteration in which $q$ invokes $\W{\bar{v},r}$ by executing line~\ref{aleader}.\RMP{Name this iteration $I_q$ and use it later?}
Process $q$ reads all the registers in line~\ref{readall} of that iteration. 
\begin{claim}\label{bc161}
$q$ is a leader and $q$ reads $(\bar{v},r-1)$ from the register of every leader in line~\ref{readall}.
\end{claim}

\begin{proof}
Since $q$ invokes $\W{\bar{v},r}$ in line~\ref{aleader} of that iteration,
    $q$ previously reads $(-,r-1)$ from $R[q]$ in line~\ref{readall}; moreover,
    by the condition of line~\ref{lagree},
    $q$ reads $(\bar{v},r')$ for some $r'$ from the register of every leader.
So, by the definition of leaders, $r' \geq r-1$.
By the definition of $q$,
    no process invokes $\W{\bar{v},r}$ before 
    $q$ invokes $\W{\bar{v},r}$,
    and so
    no process invokes $\W{\bar{v},r}$ before 
    $q$ completes reading registers in line~\ref{readall}.
By Lemma~\ref{bl2},
    no process invokes $\W{\bar{v},j}$ for any $j\geq r$ before $q$ completes reading registers in line~\ref{readall}.
Thus,
    since $q$ reads $(\bar{v},r')$ with $r' \geq r-1$ from every leader's register in line~\ref{readall},
    and the registers are regular,
    $q$ reads $(\bar{v},r-1)$ from every leader's register.
Recall that $q$ reads $(-,r-1)$ from $R[q]$ in line~\ref{readall}.
So,
    by the definition of leaders,
    $q$ is a leader.
\end{proof}
By Claim~\ref{bc161}, $q$ reads $(\bar{v},r-1)$ from $R[q]$ in line~\ref{readall}.
Since $R[q]$ is a regular register,
	$q$ previously invoked $\W{\bar{v},r-1}$.
Therefore we have:

\begin{claim}\label{bc161.5}
$q$ completes $\W{\bar{v},r-1}$ before $q$ starts reading registers in line~\ref{readall}.
\end{claim}

\begin{claim}\label{bc162}
$q$ completes $\W{\bar{v},r-1}$ before $p$ completes $\W{v,r}$.
\end{claim}

\begin{proof}
From Claim~\ref{bc161},
    $q$ reads $(-,r')$ with some $r'\leq r-1$ from $R[p]$ in line~\ref{readall}.
Since $r'\leq r-1$,
	by Observation~\ref{observation0}(\ref{mono}),
$p$ invokes $\W{-,r'}$ before invoking $\W{v,r}$.
Since $q$ reads $(-,r')$ from $R[p]$,
    $q$ starts reading registers in line~\ref{readall} before $p$ completes $\W{v,r}$.
By Claim~\ref{bc161.5},
   $q$ completes $\W{\bar{v},r-1}$ before $p$ completes $\W{v,r}$.
\end{proof}

Now consider the while iteration in which $p$ decides $v$ at round $r$.

\begin{claim}\label{bc163}
$p$ reads $(\bar{v},r')$ or $(\bot,r')$ with $r' = r-1$ or $r$ from $R[q]$ in line~\ref{readall} of the while iteration where $p$ decides $v$ at round $r$.\XMP{more clear if possible}
\end{claim}

\begin{proof}
By Claim~\ref{bc162},
    $q$ completes $\W{\bar{v},r-1}$ before $p$ starts reading registers in line~\ref{readall}.
Since $R[q]$ is a regular register,
	by Observation~\ref{observation0}(\ref{mono}),
    $p$ reads $(-,r')$ with $r' \geq r-1$ from $R[q]$ in line~\ref{readall}. 
Since $p$ decides $v$ at round $r$ in that iteration,
	(1) $p$ has $R[p]=(v,r)$ in line~\ref{readall}, and
	(2) by the condition of line~\ref{notdecide}, $p$ is a leader.
Thus, $p$ does not read $(-,j)$ for any $j > r$ from any register in line~\ref{readall}.
This implies that $p$ reads $(-,r')$ with $r' = r-1$ or $r$ from $R[q]$ in line~\ref{readall} in that iteration.

By the definition of $q$, process $q$ invokes $\W{\bar{v},r}$; by Claim~\ref{bc161.5}, $q$ completes $\W{\bar{v},r-1}$.
So, by Observation~\ref{observation0}(\ref{uno2}),
    $q$ does not invoke $\W{v,r'}$ with $r' = r-1$ or $r$.
Since $R[q]$ is a regular register,
	when $p$ reads $(-,r')$ from $R[q]$ in line~\ref{readall},
   $p$ reads $(\bar{v},r')$ or $(\bot,r')$ with $r' = r-1$ or $r$. 
\end{proof}

When $p$ executes line~\ref{notdecide} with $R[p]=(v,r)$, by Claim \ref{bc163},
	 $p$ sees that it disagrees with process $q$,
	and $q$ trails by at most one round, and so $p$ does \emph{not} decide in line~\ref{decide} --- a contradiction.
\item
\textbf{Case 2}: $q$ invokes $\W{\bar{v},r}$ by executing line \ref{acoin}.
Consider the while iteration in which $q$ invokes $\W{\bar{v},r}$ by executing line~\ref{acoin}.

\begin{claim}\label{bc163.5}
$q$ completes $\W{\bot,r-1}$ before it starts reading registers in line~\ref{readall}.
\end{claim}

\begin{proof}
Since $q$ invokes $\W{\bar{v},r}$ by executing line \ref{acoin},
    by the condition of line~\ref{checkbot},
    $q$ reads $(\bot,r-1)$ from $R[q]$ in line \ref{readall}.
Thus,
	since $R[q]$ is a regular register,
    $q$ completes $\W{\bot,r-1}$ before it starts reading registers in line~\ref{readall}.
\end{proof}

\begin{claim}\label{bc164}
$q$ reads $(-,r-1)$ from the register of every leader in line~\ref{readall}.\RMP{Like in Claim~\ref{bc161}, $q$ is leader, but we don't use this fact}\XMP{more clear if possible}
\end{claim}

\begin{proof}
Since $q$ invokes $\W{\bar{v},r}$ by executing line \ref{acoin},
    $q$ must see that the leaders do not agree on any value $v^*$ in line~\ref{lagree}.
Thus, for some $r'$,
    either $q$ reads both $(v,r')$ and $(\bar{v},r')$ from the leaders' registers in line~\ref{readall},
    or $q$ reads $(\bot,r')$ from at least one leader's register in line~\ref{readall}.
So, for some leader~$\ell$,
    $q$ reads $(\bot,r')$ or $(\bar{v},r')$ from $R[\ell]$ in line~\ref{readall}.
By Claim~\ref{bc163.5} and the definition of leaders,
    $r' \geq r-1$.
By the definition of $q$,
    no process invokes $\W{\bar{v},r}$ before
    $q$ invokes $\W{\bar{v},r}$,
    and so
    no process invokes $\W{\bar{v},r}$ before 
    $q$ completes reading registers in line~\ref{readall}.
So,
    by Lemma \ref{bl2} and Lemma \ref{bl2'},
    no process invokes $\W{\bar{v},j}$ or $\W{\bot,j}$ for any $j\geq r$ before $q$ completes reading registers in line \ref{readall}.
Since $q$ reads $(\bot,r')$ or $(\bar{v},r')$ with $r' \geq r-1$ from $R[\ell]$ in line~\ref{readall},
    $r' = r-1$.
Since $\ell$ is a leader,
    $q$ reads $(-,r-1)$ from the register of every leader in line~\ref{readall}.
\end{proof}

\begin{claim}\label{bc165}
$q$ completes $\W{\bot,r-1}$ before $p$ completes $\W{v,r}$.
\end{claim}

\begin{proof}
From Claim~\ref{bc164},
    $q$ reads $(-,r')$ with some $r'\leq r-1$ from $R[p]$ in line~\ref{readall}.
Since $r'\leq r-1$,
	by Observation~\ref{observation0}(\ref{mono}),
$p$ invokes $\W{-,r'}$ before invoking $\W{v,r}$.
Since $q$ reads $(-,r')$ from $R[p]$,
    $q$ starts reading registers in line~\ref{readall} before $p$ completes $\W{v,r}$.
By Claim~\ref{bc163.5},
    $q$ completes $\W{\bot,r-1}$ before $p$ completes $\W{v,r}$.
\end{proof}

Now consider the while iteration in which $p$ decides $v$ at round $r$.
\begin{claim}\label{bc166}
$p$ reads $(\bot,r-1)$ , $(\bar{v},r)$ or $(\bot,r)$ from $R[q]$ in line~\ref{readall} of the while iteration where $p$ decides $v$ at round $r$.
\end{claim}
\begin{proof}
By Claim \ref{bc165},
    $q$ completes $\W{\bot,r-1}$ before $p$ starts reading registers in line \ref{readall}.
Since $R[q]$ is a regular register,
	by Observation~\ref{observation0}(\ref{mono}),
    $p$ reads $(-,r')$ with $r' \geq r-1$ from $R[q]$ in line~\ref{readall}.
Since $p$ decides $v$ at round $r$ in that iteration,
	(1) $p$ has $R[p]=(v,r)$ in line~\ref{readall}, and
	(2) by the condition of line~\ref{notdecide}, $p$ is a leader.
Thus,
   $p$ does not read $(-,j)$ for any $j > r$ from any register in line~\ref{readall}.
This implies $p$ reads $(-,r')$ with $r' = r-1$ or $r'=r$ from $R[q]$ in line~\ref{readall}. 

By Observation~\ref{observation0}(\ref{due}),
    $q$ completes $\W{v,r-1}$ or $\W{\bar{v},r-1}$ \emph{before} it completes \linebreak $\W{\bot,r-1}$.
Recall that $p$ starts reading registers in line~\ref{readall}
	after it completes $\W{v,r}$.
So, by Claim~\ref{bc165}
    $p$ starts reading registers in line~\ref{readall}
    after $q$ completes $\W{\bot,r-1}$.
Thus,
	since $R[q]$ is a regular register,
    if $r' = r-1$, then $p$ must read $(\bot,r-1)$ from $R[q]$ in line~\ref{readall}.    
Since $q$ invokes $\W{\bar{v},r}$,
    by Observation~\ref{observation0}(\ref{uno2}),
    $q$ does not invoke  $\W{v,r}$.
Thus,
	since $R[q]$ is a regular register,
    if $r' = r$, then $p$ must read  $(\bar{v},r)$ or $(\bot,r)$ from $R[q]$ in line~\ref{readall}.
Therefore,
   $p$ reads $(\bot,r-1)$, $(\bar{v},r)$ or $(\bot,r)$ from $R[q]$ in line \ref{readall}. 
\end{proof}

When $p$ executes line~\ref{notdecide} with $R[p]=(v,r)$, by Claim \ref{bc166},
	 $p$ sees that it disagrees with process $q$,
	and $q$ trails by at most one round, and so $p$ does \emph{not} decide in line~\ref{decide} --- a contradiction.
\end{itemize}
Since both cases lead to a contradiction the lemma holds.
\end{proof}

\begin{lemma}[Agreement]\label{bl6}
If some process $p$ decides $v$,
    no process decides $\bar{v}$.
\end{lemma}

\begin{proof}
Assume,
    for contradiction,
   a process $p$ decides $v$ at round $r$ with $r\geq 1$ and a process $q\neq p$ decides $\bar{v}$ at round $r'$ with $r' \geq 1$.
Without loss of generality,
	assume $r' \geq r$.
Since $p$ decides $v$ at round $r$,
    by Lemma \ref{bl3.5},
    no process invokes $\W{\bar{v},r}$.
So,
by Corollary~\ref{bc2.2},
every process that invokes $\W{-,r''}$ for any $r'' \geq r$ invokes $\W{v,r''}$ (*).
Since $q$ decides $\bar{v}$ at round $r'$ in some while iteration,
	by Definition~\ref{def-decide},
    $q$ reads $(\bar{v},r')$ from $R[q]$ in line \ref{readall} of that iteration.
Since $R[q]$ is a regular register,
    $q$ invokes $\W{\bar{v},r'}$.
By (*), $ r' < r$ --- a contradiction to $r' \geq r$.
\end{proof}

We now proceed to show that the algorithm satisfies the \textbf{Termination} property with probability 1.  

\begin{observation}\label{observation1}
For all $r \ge 1$, if a correct process $p$ does not decide a value at any round $r' \le r$, then $p$ eventually invokes $\W{-,r}$.
\end{observation}


\begin{lemma}\label{bl3}
For all $r \ge 1$,
	if no process invokes $\W{\bar{v},r}$,
    then every correct process that invokes $\W{-,r+1}$ decides $v$ at round $r+1$.
\end{lemma}

\begin{proof}
Assume,
    for contradiction,
    there is a round $r \geq 1$ such that no process invokes $\W{\bar{v},r}$,
    some correct process $x$ invokes $\W{-,r+1}$,
    but $x$ does not decide $v$ at round $r+1$.
Since no process invokes $\W{\bar{v},r}$, by Corollary~\ref{bc2.2},
	$x$ invokes $\W{v,r+1}$.
By Observation~\ref{observation0}(\ref{uno2}), $x$ never invokes $\W{\bar{v},r+1}$.
So $x$ cannot decide $\bar{v}$ at round $r+1$.
Since $x$ does not decide $v$ or $\bar{v}$ at round $r+1$,
	and $x$ invokes $\W{v,r+1}$,
	$x$ must fail the condition in line~\ref{notdecide} with $(v,r+1)$ in its register, i.e., with $R[x] = (v,r+1)$.
Let $p$ be the \emph{first} process that fails the condition in line~\ref{notdecide} with $(v,r+1)$ in its register, i.e., with $R[p] = (v,r+1)$.

Process $p$ does so in line~\ref{notdecide} of some while iteration.
Since $R[p]$ is a regular register,
    when $p$ reads all the registers in line~\ref{readall} in this iteration,
    $p$ reads $(v,r+1)$ from $R[p]$.
Since $p$ fails the condition in line~\ref{notdecide},
    when $p$ reads the registers $R[-]$ in line~\ref{readall},
    $p$ sees that
    (1) $p$ is not a leader, or 
    (2) $p$ is a leader but some process that trails $p$ by less than 2 rounds disagrees with $p$;
    i.e., there is a process $y$ such that $p$ reads $(\bar{v},r')$ or $(\bot,r')$ from
    register $R[y]$ in line~\ref{readall} with $r' = r$ or $r' = r+1$.
We now prove that both cases are impossible (and so the lemma holds):

\textbf{Case 1:}
$p$ is not a leader.
Since $p$ reads $(v,r+1)$ from $R[p]$ in line~\ref{readall},
	and $p$ is not a leader,
	there must be at least one process $q \neq p$
	such that $p$ reads $(-,r'')$ with $r'' > r+1$ from $R[q]$.
Since $R[q]$ is a regular register,
    $q$ invokes $\W{-,r''}$ before $p$ completes reading $R[q]$ in line~\ref{readall}.
Thus,
	$q$ invokes $\W{-,r''}$
    before $p$ fails the condition in line~\ref{notdecide} with $R[p] = (v,r+1)$~(*).
 
Since $r'' > r+1$, $q$ must invoke $\W{-,r+1}$ before invoking $\W{-,r''}$.
Since by assumption no process invokes $\W{\bar{v},r}$, by Corollary~\ref{bc2.2},
	 $q$ invokes $\W{v,r+1}$.
So $q$ fails the condition in line~\ref{notdecide} with $R[q] = (v,r+1)$ before invoking $\W{-,r''}$ with $r'' > r+1$.
Thus, by (*) $q$ fails the condition in line~\ref{notdecide} with $R[q] = (v,r+1)$
    before $p$ fails the condition in line~\ref{notdecide} with $R[p] = (v,r+1)$.
Since $q \neq p$, this contradicts the definition of $p$.

\textbf{Case 2:}
There is a process $y$ such that $p$ reads $(\bar{v},r')$ or $(\bot,r')$ from
    register $R[y]$ in line~\ref{readall} with $r' = r$ or $r' = r+1$.
Since $R[y]$ is a regular register,
    process $y$ invokes $\W{\bar{v},r'}$ or $\W{\bot,r'}$ with $r' = r$ or $r' = r+1$.
This contradicts Corollary~\ref{bc2.2} since, by assumption,
    no process invokes $\W{\bar{v},r}$.
\end{proof}

\begin{lemma}\label{bl4}
Suppose some process completes a $\W{-,r-1}$ operation for some round $r \ge 2$.
Let $\W{v,r-1}$ be the \emph{first} $\W{-,r-1}$ operation that completes.\footnote{Note that from Observation~\ref{observation0}(\ref{due}), $v \not = \bot$.}
If all the processes that invoke a $\W{flip(),r}$ operation do so with $flip()=v$,
	then no process invokes $\W{\bar{v},r}$.
 \end{lemma}

\begin{proof}
Suppose some process completes a $\W{-,r-1}$ operation for some round $r \ge 2$.
Let $\W{v,r-1}$ be the \emph{first} $\W{-,r-1}$ operation that completes,
	and $p$ be the process that completes this $\W{v,r-1}$ operation.
Assume that some process invokes $\W{\bar{v},r}$;
	we now prove that
    some process invokes $\W{flip(),r}$ operation with $flip()= \bar{v}$,
    and so the lemma holds.
    
Let $q$ be the \emph{first} process that invokes $\W{\bar{v},r}$;
note that $q$ could be $p$. 
According to the algorithm,
	$q$ invokes $\W{\bar{v},r}$ by executing either (1) line~\ref{aleader}
	or (2) line~\ref{acoin}.
So there are two cases:
\begin{itemize}
\item  
\textbf{Case 1:} $q$ invokes $\W{\bar{v},r}$ by executing line~\ref{acoin}.
In this case, $q$ invokes $\W{flip(),r}$ with $flip() = \bar{v}$, as we want to show.

\item
\textbf{Case 2:} $q$ invokes $\W{\bar{v},r}$ by executing line~\ref{aleader}.
Consider the while iteration in which $q$ invokes $\W{\bar{v},r}$ by executing line~\ref{aleader}.

\begin{claim}\label{9c3}
$q$ reads $(-,r'')$ with $r'' \geq r-1$ from $R[p]$ in line~\ref{readall}.
\end{claim}

\begin{proof}

Since $q$ invokes $\W{\bar{v},r}$ in line~\ref{aleader},
	it is clear that $q$ starts reading registers in line~\ref{readall} after it completes its $\W{-,r-1}$ operation.
Since $p$ is the \emph{first} process that completes a $\W{-,r-1}$ operation,
    $q$ starts reading registers in line~\ref{readall} after $p$ completes its $\W{v,r-1}$ operation.
Since $R[p]$ is a regular register,
	by Observation~\ref{observation0}(\ref{mono}),
    $q$ reads $(-,r'')$ with $r'' \geq r-1$ from $R[p]$ in line~\ref{readall}.
\end{proof}

Let $r'$ be the round of the leaders that $q$ sees in line~\ref{readall},
	i.e., $q$ reads $(-,r')$ from all the leaders' registers in line~\ref{readall}.
By Claim~\ref{9c3} and the definition of leaders, $r' \geq r-1$.
Since $q$ invokes $\W{\bar{v},r}$ in line~\ref{aleader},
    the leaders that $q$ sees agree on $\bar{v}$; more precisely,
    $q$ reads $(\bar{v} , r')$ from all the leaders' registers.
 
 \begin{claim}\label{9c2}
$r' = r-1$.
\end{claim}

\begin{proof}
Since $q$ completes reading all the registers in line~\ref{readall}
    before it invokes $\W{\bar{v},r}$ in line~\ref{aleader},
    by the choice of $q$,
    $q$ completes reading the registers before any process invokes $\W{\bar{v},r}$.
In other words, no process invokes $\W{\bar{v},r}$
	before $q$ completes reading the registers in line~\ref{readall}.
By Lemma~\ref{bl2},
	no process invokes $\W{\bar{v},j}$ with $j \ge r$
	before $q$ completes reading the registers in line~\ref{readall}.
Thus,
   since the registers are regular,
    $q$ does not read $(\bar{v},j)$ with $j\geq r$ from any register in line~\ref{readall}.
Since $q$ reads $(\bar{v},r')$ with $r' \geq r-1$ from the leaders' registers,
	$r' = r-1$.
\end{proof}

From Claim \ref{9c2},
    $q$ reads $(\bar{v}, r-1)$ from the register of every leader in line~\ref{readall}.
So, by the definition of leaders and
    Claim~\ref{9c3},
   $r'' = r-1$ and so $p$ must be one of the leaders that $q$ sees,
    i.e., $q$ reads $(\bar{v}, r-1)$ from $R[p]$ in line~\ref{readall}.
Since $R[p]$ is a regular register,
    $p$ invokes $\W{\bar{v},r-1}$.
Thus,
    by the definition of $p$,
    $p$ invokes both $\W{v,r-1}$ and $\W{\bar{v},r-1}$ --- a contradiction to Observation~\ref{observation0}(\ref{uno2}).
Therefore,
    Case 2 is impossible. 
\end{itemize}
Since  in Case 1, process $q$ invokes $\W{flip(),r}$ with $flip() = \bar{v}$ and Case 2 is impossible,
the lemma holds.
\end{proof}

\begin{lemma}\label{bc1}
Suppose some process completes a $\W{-,r-1}$ operation for some round $r \ge 2$.
Let $\W{v,r-1}$ be the \emph{first} $\W{-,r-1}$ operation that completes.
If all the processes that invoke a $\W{flip(),r}$ operation do so with $flip()=v$,
     then every correct process decides at some round $r' \leq r+1$. 
\end{lemma}

\begin{proof}
Suppose some process completes a $\W{-,r-1}$ operation for some round $r \ge 2$, and
	let $\W{v,r-1}$ be the \emph{first} $\W{-,r-1}$ operation that completes.
Since all the processes that invoke a $\W{flip(),r}$ operation do so with $flip()=v$,
    by Lemma \ref{bl4},
    no process invokes $\W{\bar{v},r}$.
Then by Lemma \ref{bl3},
    every correct process that invokes $\W{-,r+1}$ decides $v$ at round $r+1$.
By Observation~\ref{observation1}, every correct process that does not
	invoke $\W{-,r+1}$ must decide at some round $r' \le r+1$.
So every correct process decides at some round $r' \leq r+1$.
\end{proof}

%

\begin{lemma}[Termination]\label{termination}
The algorithm terminates with probability 1, even against a strong adversary.\RMP{Say this better }
\end{lemma}

\begin{proof}
Consider any round $r \ge 2$ such that some correct processes have not yet decided, i.e., they have not decided at any round $r' \le r$.
Thus, some process has not decided at any round $r' \le r-1$.
By Observation~\ref{observation1},
	this process eventually invokes $\W{-,r-1}$, and since it is correct, it also completes $\W{-,r-1}$.
Suppose the \emph{first} $\W{-,r-1}$ operation that completes is $op= \W{v,r-1}$.
Since any invocation of $\W{flip(),r}$ can occur only \emph{after} $op$ completes,
	the value $v$ is set \emph{before} any invocation of $\W{flip(),r}$, i.e., before any coin flip at round $r$.
Since a strong adversary cannot control coin flips,
    it cannot prevent the coin flips at round $r$ to match the value $v$ that was set before the first coin toss at round $r$.
Thus,
    there is a positive probability $\epsilon \ge 2^{-n}$
    that all the processes that invoke a $\W{flip(),r}$ operation do so with $flip()=v$.
 So, by Lemma \ref{bc1},
   with probability $\epsilon$ every correct process decides at some round $r' \leq r+1$;
    i.e.,
    with probability $\epsilon$ the algorithm terminates by round $r+1$.
Since this holds for every round $r \ge 2$ such that some correct processes have not decided at any round $r' \le r$,
	the algorithm terminates with probability $1$, even against a strong adversary.
\end{proof}

By Lemmas~\ref{validity},~\ref{bl6}, and~\ref{termination}, the algorithm satisfies Validity, Agreement, and Termination with probability 1.
So we have:

\begin{theorem}
The randomized consensus algorithm of Aspnes and Herlihy shown in Algorithm~\ref{Algocon} works with
	regular SWMR registers against a strong adversary.
\end{theorem}

\section*{Acknowledgements}
We thank Kevan M. Hollbach for his helpful comments on this paper.

\bibliographystyle{abbrv}
\bibliography{consen.bib}
   
\appendix
\counterwithin{NewCounter}{section}
\section{Appendix}

We now briefly explain how the algorithm in~\cite{AH1990consensus} was shown to terminate with probability 1
	under the assumption that registers are \emph{atomic} (where each operation is instantaneous),
	and why this proof does not hold if we replace them with \emph{linearizable}
	(implementations of) registers (where each operation spans an interval).

The following Lemma, shown in~\cite{AH1990consensus}, is central to the algorithm's original proof of termination:

\begin{lemma}\label{l7}
Let $v$ be the first value written at round $r-1$.
If all the processes that flip a coin at round $r$ get the value $v$,\footnote{This means every process that writes $(flip(),r)$ to its atomic register does so with $flip()=v$.} 
then all processes have the same preference at round $r$.
\end{lemma}
%

The termination argument  given in~\cite{AH1990consensus} uses Lemma~\ref{l7}, and intuitively it goes as follows.
Before any process flips a coin at any round $r$, at least one process writes some value at round $r-1$. 
Since registers are atomic, the \emph{first value} written by a process at round $r-1$, say value $v$,
	is \emph{fixed before} any process flips a coin at round~$r$.
Thus, even a strong adversary cannot prevent any process that flips a coin at round~$r$
	to get the value $v$ that was determined before any of these coin tosses.
So all the processes that flip a coin at round $r$ get the value $v$
	with \emph{some positive probability $\epsilon$} ($\epsilon \ge 2^{-n}$ if the system has $n$ processes).
The above argument is the core of proof of termination with probability 1,
	because by Lemma~\ref{l7}, if this happens at any round $r$
	then all processes have the \emph{same} preference at round $r$,
	and (as shown in~\cite{AH1990consensus}) this causes the algorithm to terminate by round $r+2$.

When the algorithm's atomic registers are replaced with linearizable registers,
	the above argument, however, fails:
	in a nutshell, even though Lemma~\ref{l7} still holds,
	it can no longer be used to prove that the algorithm terminates.
Intuitively, this is because 
	it is no longer true that
	the first value written by a process at round $r-1$
	is fixed before any process flips a coin at round~$r$:
	as we explain below,
	with linearizable registers,
	a stronger adversary can \emph{first}
	see the result of some coin toss at round $r$,
	and \emph{then} control which write operation is linearized \emph{first} among the write operations of round $r-1$.

\begin{figure}[!htb]
    \centering 
    \captionsetup{justification=centering}
    \includegraphics[width=0.6\textwidth]{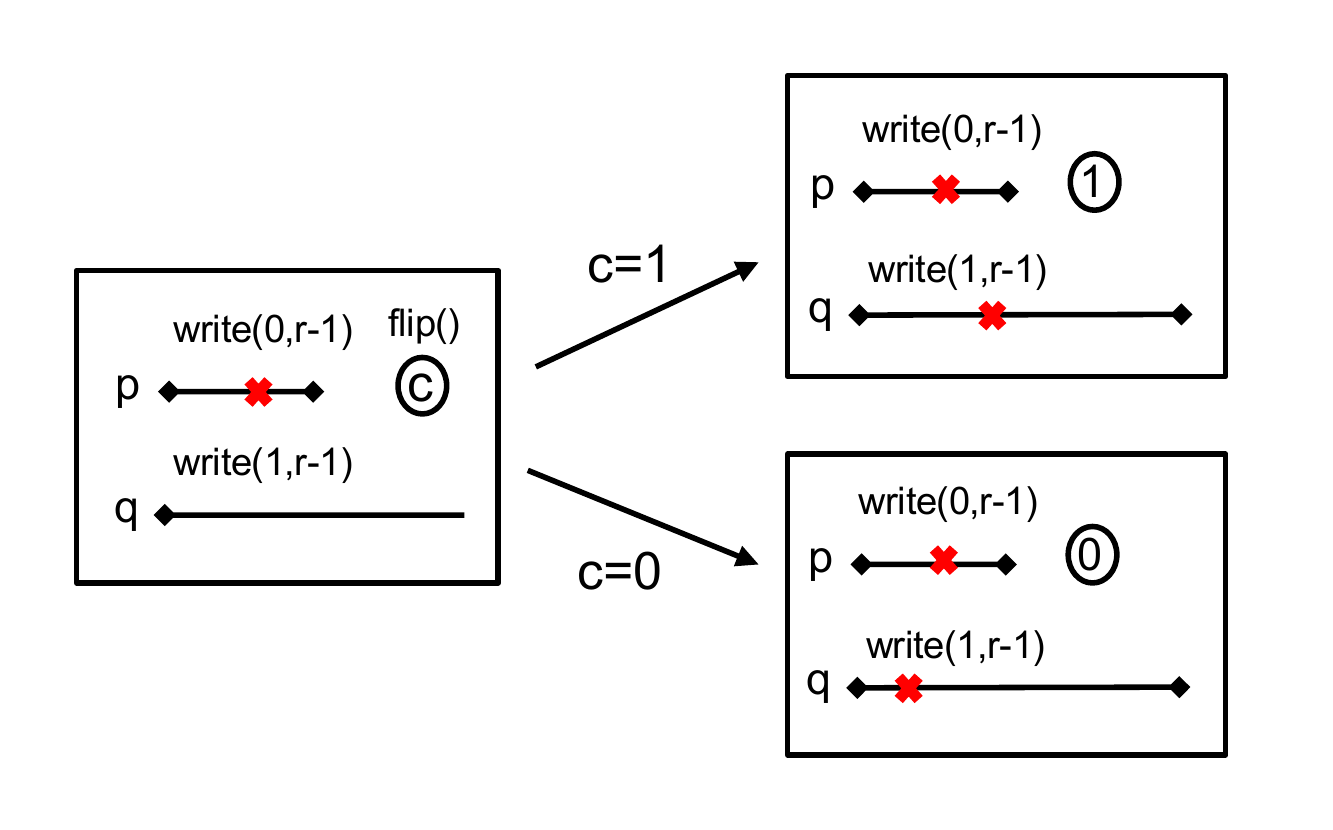}
    \caption{Adversary controlling which value is the first one to be written in round $r-1$} 
    \label{appendixfig}
\end{figure}

 With linearizable registers, every register operation spans an \emph{interval} that starts with an invocation and is followed by a response.
A strong adversary can take advantage of this by scheduling processes and operations as follows (see Figure~\ref{appendixfig}):

\begin{enumerate}[(a)]

\item some processes invoke $\W{0,r-1}$ and some processes invoke $\W{1,r-1}$,
	and all the processes invoke these $\W{-,r-1}$ operations at the same time;\XMP{at the same time -> one-after-the-other/precede all responses.}

\item one of these operations completes, say the operation $w_p = \W{0,r-1}$ by process $p$;

\item after completing the operation $w_p$, $p$ sees both $(0,r-1)$ and $(1,r-1)$ from the registers that it reads in line~\ref{readall},
		and so $p$ eventually flips a coin in line~\ref{acoin} and gets some value $c \in \{0,1\}$ at round $r$;
	
\item if $p$ gets $c=0$, then the adversary schedules one of the concurrent $\W{1,r-1}$
	operations (namely the write by process $q$ in Figure~\ref{appendixfig})
	so that it is the write that is linearized \emph{first} among all the $\W{-,r-1}$ operations;
	if $p$ gets $c=1$, then the adversary schedules $w_p = \W{0,r-1}$ to be
	the write that is linearized \emph{first} among all the $\W{-,r-1}$ operations.
	
So in both cases, the adversary can schedule operations to ensure that the first value written
	at round $r-1$ is \emph{not} equal to the coin that $p$ flips at round $r$.

\end{enumerate}

%
%
%
%
%
%

Thus the core argument of the proof of termination given in~\cite{AH1990consensus}
	for atomic registers does not hold when registers are replaced with linearizable (implementations of) registers:
	even though Lemma~\ref{l7} still holds, the adversary can schedule processes
	and operations so that it can never be applied.

Finally, it is also worth noting that if registers are \emph{strongly linearizable}~\cite{sl1}, then the original proof of termination given in \cite{AH1990consensus} holds again: by the time any process flips a coin in line~\ref{acoin} at round $r$, the $\W{-,r-1}$ operation that is linearized first has already been determined.
\end{document}